\documentclass[letterpaper, 10 pt, conference]{ieeeconf}  

\IEEEoverridecommandlockouts                              
\usepackage{amsfonts,amssymb}  
\usepackage[cmex10]{amsmath}
\usepackage[pdftex]{graphicx}
\usepackage{array,cite,url}
\usepackage[font=footnotesize,caption=false]{subfig}

\newtheorem{thm}{Theorem}[section]
\newtheorem{lm}[thm]{\bf Lemma}

\newtheorem{df}[thm]{\bf Definition}

\newtheorem{rem}[thm]{\bf Remark}

\usepackage{pgf}
\usepackage{tikz}
\usetikzlibrary{arrows,automata,calc}
\usetikzlibrary{positioning}

\tikzset{
    state/.style={
           rectangle,
           rounded corners,
           draw=black, very thick,
           minimum height=2em,
           inner sep=1pt,
           text centered,
           },
		normal/.style={
           rectangle,
           rounded corners,
           draw=black, thin,
           minimum height=2em,
           inner sep=1pt,
           text centered,
           },
		every picture/.style={font issue=\footnotesize},
    font issue/.style={execute at begin picture={#1\selectfont}}
}

\title{\LARGE \bf
Timed Automata Approach for Motion Planning Using Metric Interval Temporal Logic
}

\author{Yuchen Zhou, Dipankar Maity and John S. Baras
\thanks{This work is supported by NSF grant CNS-1035655, NIST grant 70NANB11H148 and by DARPA (through ARO) grant W911NF1410384. The authors are with the Department of Electrical and Computer Engineering, and the Institute for Systems Research, University of Maryland, College Park, Maryland, USA. email: {\tt\small \{yzh89, dmaity, baras\}@umd.edu}.}
}

\begin{document}

\maketitle
\thispagestyle{empty}
\pagestyle{empty}

\begin{abstract}
In this paper, we consider the robot motion (or task) planning problem under some given time bounded high level specifications. We use metric interval temporal logic (MITL), a member of the temporal logic 
family, to represent the task specification and then we provide a constructive way to generate a timed automaton and methods to look for accepting runs on the automaton to find a feasible motion (or path) 
sequence for the robot to complete the task. 
\\
\end{abstract}


\section{Introduction}
Motion planning and task planning have gained  an enormous thrust in the robotics community in the past decade or so. Though, motion (task) planning has attracted a great deal of research in the past few decades, however recently, researchers have come up with new metrics and methodology to represent motion and task specifications. Initially, motion planning for a mobile robot started with the aim of moving a point mass from an initial position to a final position in some optimal fashion. With course of time, people started to consider planning in cluttered domains (i.e. in presence of obstacles) and also accounted for the dimensionality and the physical constraints of the robot. 

Though we have efficient approaches for general motion planning, very few are available or scalable to plan in dynamic environments or under finite time constraints. Temporal logics have been used greatly 
to address complex motion specifications, motion sequencing and timing behaviors etc. Historically temporal logic was originated for model checking, validation and verification in software community \cite{baier2008}
and later on researchers found it very helpful to use Linear Temporal logic (LTL), Computational Tree logic (CTL), Signal Temporal logic (STL) etc. for representing complex motion (or task) specifications. 
The developments of tools such as SPIN \cite{SPIN}, NuSMV \cite{NuSMV} made it easier to check if a given specifications can be met by creating a suitable automaton and looking 
for a feasible path on that automaton. However, the construction of the automaton from a given temporal logic formula is based on the implicit assumption that there is no time constraints associated with the specification.

Currently motion planning for robots is in such a stage where it is very crucial to incorporate time constraints since these constraints can arise from different aspects of the problem: dynamic environment, sequential 
processing, time optimality etc. Planning with time bounded objectives is inherently hard due to the fact that every transition from one state to another in the associated automaton has to be carried out, 
by some controller, exactly in time from an initial configuration to the final configuration. Time bounded motion planning has been done in heuristic ways \cite{kant, Erdmann} and also by using 
mixed integer linear programming (MILP) framework \cite{richards, zhou}. In this paper, we are interested in extending the idea of using LTL for time-unconstrained planning to use MITL for time-constrained 
motion planning. In \cite{maity}, the authors proposed a method to represent time constrained planning task as an LTL formula rather than MITL formula. This formulation reduced the complexity of 
\textsc{Exp-space}-complete for MITL to \textsc{Pspace}-complete for LTL. However, the number of states in the generated B\"{u}chi automata increases with time steps. 

In this paper, we mainly focus on motion planning based on the construction of an efficient timed automaton from a given MITL specification. 
A dedicated controller to navigate the robot can be constructed for the general planning problem
once the discrete path is obtained from the automaton.  The earlier results on construction of algorithms to verify timing properties of real 
time systems can be found in \cite{Alur1996}. The complexity of satisfiability and model checking problems for MTL formulas has been already studied in \cite{SurveyOuaknine08} and it has been shown that commonly 
used real-time properties such as bounded response and invariance can be decided in polynomial time or exponential space. More works on the decidability on MTL can be found in \cite{MTLOuaknine05} 
and the references there in. The concept of alternating timed automata for bounded time model checking can be 
found in \cite{Jenkins2010}. \cite{Nickovic2010} talks about constructing deterministic timed automata from MTL specifications and this provides a unified framework to include all the future operators
of MTL. The key to the approach of \cite{Nickovic2010} was in separating the continuous time monitoring from the discrete time predictions of the future. We restrict our attention to generate timed automata 
from MITL based on the work done in \cite{Maler2006a}. It is done by constructing a timed automaton to generate a sequence of states and another to check whether the sequence generated is actually a valid one in the sense that it satisfies the given MITL specification.

The rest of the paper is organized as follows, section \ref{sec:pre} provides a background on MITL and the timed automata based approach for MITL. Section \ref{sec:motion_planning} illustrates how the timed automata can be used to motion synthesis and we also provide UPPAAL \cite{uppaal} implementation of the same. Section \ref{sec:case} gives some examples on different time bounded tasks and shows the implementation results. Section \ref{sec:continuous} provides a brief overview of how a continuous trajectory can be generated from the discrete plan. Finally, we conclude in section \ref{sec:conclusion}.

\section{Preliminaries} \label{sec:pre}

In this paper, we consider a surveying task in an area by a robot whose motion is abstracted to a graph. In particular for our particular setup, the robot motion is captured as a timed automaton (Fig. \ref{fig:map}). Every edge is a timed transition that represents navigation of the robot from one location to other in space and every vertex of the graph represents a partition of the space. Our objective is to find an optimal time path that satisfies the specification given by timed temporal logic.

\begin{figure}%
\centering
\begin{tikzpicture}[->,>=stealth']

 \node[normal] (S0) 
 {\begin{tabular}{l}
  pos0\\
	$z\leq 1$\\
 \end{tabular}};
  
 \node[normal,    	
  right of=S0, 	
  node distance=4cm, 	
  anchor=center] (S1) 	
 {\begin{tabular}{l}
  pos1:$B$\\
	$z\leq 1$\\
 \end{tabular}
 };
 
 \node[normal,
  below of=S0,
  yshift=-1cm,
  anchor=center] (S3) 
 {\begin{tabular}{l}
  pos3:$A$\\
	$z\leq 1$\\
 \end{tabular}
 };

 \node[normal,
  right of=S3,
  node distance=4cm,
  anchor=center] (S2) 
 {\begin{tabular}{l}
  pos2\\
	$z\leq 1$\\
 \end{tabular}
 };

 \path 
	(S0) 	edge  node[above]{$z\geq 1 | z:=0$} (S1)
	(S1)	edge (S0)
  (S1)  edge  node[left,align=center]{$z \geq 1$\\ $z:=0$} (S2)
	(S2) 	edge  (S1)
	(S2)	edge	node[below]{$z\geq 1 | z=0$} (S3)
	(S3)  edge  (S2)
	(S0) 	edge  node[left=0cm,align=center]{$z>0$\\$z:=0$} (S3)
	(S3)	edge	 (S0);
 
\end{tikzpicture}
\caption{Timed Automata based on cell decomposition and robot dynamics}
\label{fig:map}
\end{figure}
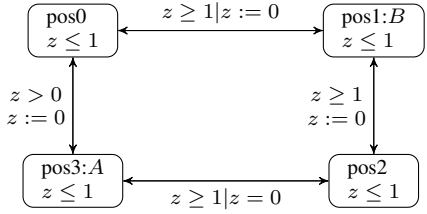

 \subsection{Metric Interval Temporal Logic (MITL)}

Metric interval temporal logic is a specification that includes timed temporal specification for model checking. It differs from Linear Temporal Logic on the part that it has constraints on the temporal operators.

The formulas for LTL are build on atomic propositions by obeying the following grammar.
\begin{df} \label{defLTL}
 \textit{The syntax of LTL formulas are defined according to the following grammar rules:}
 \begin{center}
 $\phi ::= \top ~| ~\pi~ |~\neg \phi~ | ~\phi \vee \phi ~| ~\mathbf{X} \phi | ~\phi \mathbf{U} \phi  ~ $
 \end{center}
 \end{df}
$\pi \in \Pi$ the set of propositions, $\top$ and $\bot(=\neg\top)$ are the Boolean constants $true$ and $false$ respectively.  $\vee$ denotes the disjunction operator and $\neg$ denotes the negation operator. $\mathbf{U}$ represents the Until operator. MITL extends the Until operator to incorporate timing constraints.

\begin{df} \label{def1}
 \textit{The syntax of MITL formulas are defined according to the following grammar rules:}
 \begin{center}
 $\phi ::= \top ~| ~\pi~ |~\neg \phi~ | ~\phi \vee \phi ~|~\phi \mathbf{U}_I \phi  ~ $
 \end{center} 
 \end{df}
 where $I\subseteq [0, \infty]$ is an interval with end points in $\mathbb{N} \cup \{\infty\}$.  $\mathbf{U}_I$ symbolizes the timed Until operator. Sometimes we will represent $\mathbf{U}_{[0, \infty]}$ by $\mathbf{U}$. 
Other Boolean and temporal operators such as conjunction ($\wedge$),  eventually within $I$ ($\Diamond_I$), always on $I$ ($\Box_I$) etc. can be represented using the grammar desired in definition \ref{def1}. For example, we can express time constrained eventually operator $\Diamond_I\phi \equiv \top \mathbf{U}_I\phi$ and so on. In this paper all the untimed temporal logic is transformed into until operator and all the timed operator is transformed to eventually within $I$, to make it easier to generate a timed automaton.

MITL is interpreted over $n$-dimensional Boolean $\omega$-sequences of the form $\xi: \mathbb{N} \rightarrow \mathbb{B}^n$, where $n$ is the number of propositions.
\begin{df}\label{ltlsym}
 \textit{The semantics of any MTL formula $\phi$ is recursively defined over a trajectory $(\xi, t)$ as:\\
 $(\xi, t) \models \pi$ iff $(\xi, t)$ satisfies $\pi$ at time $t$\\
 $(\xi, t) \models \neg \pi$ iff $(\xi, t)$ does not satisfy $\pi$ at time $t$\\
 $(\xi, t) \models \phi_1\vee \phi_2$ iff $(\xi, t) \models \phi_1$ or $(\xi, t) \models \phi_2$\\
 $(\xi, t) \models \phi_1\wedge \phi_2$ iff $(\xi, t) \models \phi_1$ and $(\xi, t) \models \phi_2$\\
 $(\xi, t) \models \bigcirc \phi$ iff $(\xi, t+1) \models \phi$ \\
 $(\xi, t) \models \phi_1\mathbf{U}_I \phi_2$ iff $\exists s \in I$ s.t. $(\xi, t+s) \models  \phi_2$ and $\forall$ $s' \leq s, ~ (\xi, t+s') \models \phi_1$.}

\end{df}
Thus, the expression $\phi_1 \mathbf{U}_I \phi_2$ means that $\phi_2$ will be true within time interval $I$ and until $\phi_2$ becomes true, $\phi_1$ must be true. 
The MITL operator $\bigcirc \phi$ means that the specification $\phi$ is true at next time instance, $\Box_I \phi$  means that $\phi$ is always true for the time duration $I$, 
$\Diamond_I \phi$ means that $\phi$ will eventually become true within the time interval $I$. Composition of two or more MITL operators can express very sophisticated 
specifications; for example $\Diamond_{I_1} \Box_{I_2} \phi$ means that within time interval $I_1$, $\phi$ will be true and from that instance it will hold true always for a duration
of $I_2$. Other Boolean operators such as implication ($\Rightarrow$) and equivalence ($\Leftrightarrow$) can be expressed using the grammar rules and semantics given in definitions 
\ref{def1} and \ref{ltlsym}. More details on MITL grammar and semantics can be found in \cite{MTL}, \cite{Alur1996}.  

\subsection{MITL and Timed Automata Based Approach}
An LTL formula can be transformed into a B\"{u}chi automaton which can be used in optimal path synthesis \cite{Smith2010} and automata based guidance \cite{wolff_automatonguided_2013}. Similarly, in this paper, we focus on developing a timed automata based approach for MITL based motion planning. MITL, a modification of Metric Temporal Logics (MTL), disallows the punctuation in the temporal interval, so that the left boundary and the right boundary have to be different. 
In general the complexity of model checking for MTL related logic is higher than that of LTL. The theoretical model checking complexity for LTL is \textsc{Pspace}-complete \cite{Sistla1985}. The algorithm that has been implemented is exponential to the size of the formula. MTL by itself is undecidable. 
The model checking process of MITL includes transforming it into a timed automaton \cite{Alur1996}\cite{Maler2006a}. CoFlatMTL and BoundedMTL defined in \cite{CoFlatMTLBouyer08} are more expressive fragments of MTL than MITL, which can be translated to LTL-Past but with exponential increase in size. SafetyMTL \cite{MTLOuaknine05} and MTL, evaluated over finite and discrete timed word, can be translated into alternative timed automata. Although theoretically, the results suggest many fragments of MTL are usable, many algorithms developed for model checking are based on language emptiness check, which are very different from the control synthesis i.e. finding a feasible path. From best of our knowledge, the algorithm that is close to implementation for motion planning is that of \cite{Maler2006a}.

This paper uses the MITL and timed automaton generation based on \cite{Maler2006a}. In the following section, the summary of the transformation and our implementation for control synthesis are discussed.

\section{MITL for Motion Planning}\label{sec:motion_planning}

\subsection{MITL to Timed Automata Transformation}
Consider the following requirements: a robot has to eventually visit an area $A$ and another area $B$ in time interval $[l,r]$, and the area $A$ has to be visited first. This can be captured in the following MITL,
\[
\phi = (\neg B \mathbf{U} A)\wedge (\Diamond_{[l,r]} B)
\] 

It can be represented by a logic tree structure, where every node that has children is a temporal logic operator and every leaf node is an atomic proposition, as shown in Fig. \ref{fig:tree}. Every link represents an input output relationship.

\begin{figure}
\centering
\begin{tikzpicture}[
level 1/.style={sibling distance=30mm},level 2/.style={sibling distance=20mm}, level distance=30pt,
edge from parent path={
(\tikzparentnode) |-   
($(\tikzparentnode)!0.5!(\tikzchildnode)$) -| 
(\tikzchildnode)}]                            

  \node[normal, text width=1cm] {$\wedge$}
  child {node[normal,text width=1cm] {$\mathbf{U}$}
	  child {node[normal,text width=1cm] {$\neg$}
			child {node[normal,text width=1cm]{$p(B)$}}
		}
    child {node[normal,text width=1cm] {$p(A)$}}
		}
  child {node[normal,text width=1cm] {$\Diamond_{[l,r]}$}
    child {node[normal,text width=1cm] {$p(B)$}}
  };
\end{tikzpicture}
\caption{Logic tree representation of $\phi$.}
\label{fig:tree}
\end{figure}
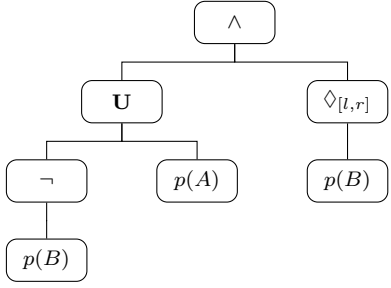

The authors in \cite{Maler2006a} propose to change every temporal logic operator into a timed signal transducer, which is a temporal automaton that accepts input and generates output. Based on their definition the Input Output Timed Automaton (IOTA) used in this paper is defined as the following to fit the control synthesis problem,

\begin{df}[Input Output Timed Automaton]\label{df:iota}
\textit{An input output timed automaton is a tuple $\mathcal{A} =(\Sigma, Q, \Gamma, \mathcal{C}, \lambda, \gamma, I, \Delta, q_0, F)$, where \\
$\Sigma$ is the input alphabet, $Q$ is the finite set of discrete states, \\
$\Gamma$ is the output alphabet, $\mathcal{C}$ is the set of clock variables, and \\
$I$ is the invariant condition defined by conjunction of inequalities of clock variables. The clock variables can be disabled and activated by setting the rate of the clock $0$ or $1$ in the invariant $I$. \\
$\lambda : Q \rightarrow \Sigma$ is the input function, which labels every state to an input, while \\
$\gamma : Q \rightarrow \Gamma$ is the output function, which labels every state to an output. \\
$\Delta$ is the transition relationship between states which is defined by $(p,q,r,g)$, where $p$ is the start state, $q$ is the end state, $r$ is the clock resets, and $g$ is the guard condition on the clock variables.\\
 $q_0$ is the initial state of the timed automaton. \\
$F$ is the set of B\"{u}chi states that have to be visited infinitely often.}
\end{df}

The transformation of Until operator and timed Eventually operator is summarized in  Figs. \ref{fig:pUq}, \ref{fig:EI_GEN} and \ref{fig:EI_CHK}. This is based on \cite{Maler2006a} with minor changes to match with our definition of IOTA. In Fig. \ref{fig:pUq}, the timed automaton for $p \mathcal{U} q$ is shown. The inputs outputs of the states are specified in the second line within the box of each state. $p\bar{q}$ means the inputs are $[1,0]$ and $\bar{p}$ means the inputs can be $[0,1]$ or $[0,0]$, and $\gamma=1$ means the output is 1. Transitions are specified in the format of $g|r$. In this case, all the transitions have guard $z>0$ and reset clock $z$. All states in this automaton are B\"{u}chi accepting states except $s_{p\bar{q}}$. The B\"{u}chi accepting states are highlighted.

\begin{figure}%
\centering
\begin{tikzpicture}[->,>=stealth']

 \node[state,
	text width=2cm] (S0) 
 {\begin{tabular}{l}
  $s_{\bar{p}}$\\
	\quad $\bar{p}/\gamma=0$\\
 \end{tabular}};
  
 \node[state,    	
  right of=S0, 	
  node distance=5cm, 	
  anchor=center,
	text width=2cm] (S1) 	
 {%
 \begin{tabular}{l} 	
  $\bar{s}_{p\bar{q}}$\\
	\quad $p\bar{q}/\gamma=0$\\
 \end{tabular}
 };
 
 \node[normal,
  below of=S0,
  yshift=-2cm,
  anchor=center,
  text width=2cm] (S2) 
 {%
 \begin{tabular}{l}
  $s_{p\bar{q}}$\\
  \quad $p\bar{q}/\gamma=1$\\
 \end{tabular}
 };

 \node[state,
  right of=S2,
  node distance=5cm,
  anchor=center] (S3) 
 {%
 \begin{tabular}{l}
  $s_{pq}$\\
  \quad $pq/\gamma=1$\\
 \end{tabular}
 };

 \path 
	(S0.5) 	edge  node[above]{$z>0 | z:=0$} (S1.175)
	(S1.185)	edge	node[below]{$z>0 | z:=0$} (S0.355)
  (S0.270)     	edge  node[left,align=center]{$z>0$\\ $z:=0$} (S2.90)
	(S2.5) 	edge  node[above]{$z>0 | z:=0$} (S3.175)
	(S3.185)	edge	node[below]{$z>0 | z:=0$} (S2.355)
	(S3.90)     	edge  node[right,align=center]{$z>0$\\$z:=0$} (S1.270)
	(S0.320) 	edge  node[left=0cm,align=center]{$z>0$\\\quad\quad$z:=0$} (S3.160)
	(S3.150)	edge	node[right=-0.2cm,align=center]{$z>0$\\\quad\quad$z:=0$} (S0.332);

\end{tikzpicture}

\caption{The timed automaton for $p \mathcal{U} q$. The inputs and outputs of the states are specified in the second line of each state. $p\bar{q}$ means the inputs are $[1,0]$ and $\bar{p}$ means the inputs can be $[0,1]$ or $[0,0]$, and $\gamma=1$ means the output is 1. Transitions are specified in the format of guard$|$reset. In this case all the transitions have guard $z>0$ and reset clock $z$. All states in this automaton are B\"{u}chi accepting states except $s_{p\bar{q}}$. The B\"{u}chi accepting states are highlighted.}%
\label{fig:pUq}%
\end{figure}
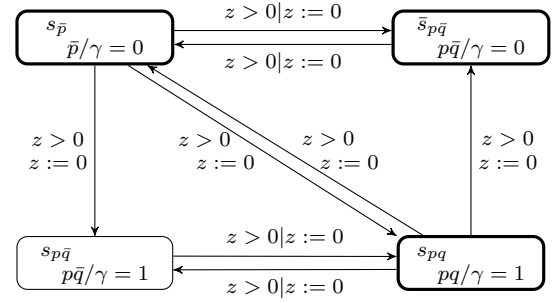

\begin{figure}%
\centering
\begin{tikzpicture}[->,>=stealth']

 \node[normal,
	text width=2cm] (S0) 
 {\begin{tabular}{l}
  $\text{Gen}_1$\\
	\quad $x_1'==1$\\
	\quad $*/\gamma=0$\\
 \end{tabular}};

 \node[normal,
	text width=6cm,
	right of=S0,
	node distance=2cm,
	yshift=1.5cm] (Init) 
 {\begin{tabular}{l}
  $\text{Gen}_0$\\
	\quad $x_i'==0$, $y_i'==0$, $\forall i=1,\ldots, m$\\
 \end{tabular}};
  
 \node[normal,    	
  right of=S0, 	
  node distance=4cm, 	
  anchor=center,
	text width=2cm] (S1) 	
 {%
 \begin{tabular}{l} 	
  $\text{Gen}_2$\\
	\quad $y_1'==1$\\
	\quad $*/\gamma=1$\\
 \end{tabular}
 };
 
 \node[normal,
  below of=S0,
  yshift=-0.5cm,
  anchor=center,
  text width=2cm] (S2) 
 {%
 \begin{tabular}{l}
  $\text{Gen}_3$\\
	\quad $x_2'==1$\\
	\quad $*/\gamma=0$\\
 \end{tabular}
 };

 \node[normal,
  right of=S2,
  node distance=4cm,
  anchor=center] (S3) 
 {%
 \begin{tabular}{l}
  $\text{Gen}_4$\\
	\quad $y_2'==1$\\
	\quad $*/\gamma=1$\\
 \end{tabular}
 };

\node[state] (Sdots) [below of=S2,draw=none,node distance=1cm] {$\ldots$};
\node[state] (Sdots2) [below of=S3,draw=none,node distance=1cm] {$\ldots$};

 \node[normal,
  below of=Sdots,
  yshift=0.1cm,
  anchor=center,
  text width=2cm,
	draw=black, thin] (S4) 
 {%
 \begin{tabular}{l}
  $\text{Gen}_{2m-1}$\\
	\quad $x_m'==1$\\
	\quad $*/\gamma=0$\\
 \end{tabular}
 };

 \node[normal,
  right of=S4,
  node distance=4cm,
  anchor=center] (S5) 
 {%
 \begin{tabular}{l}
  $\text{Gen}_{2m}$\\
	\quad $y_m'==1$\\
	\quad $*/\gamma=1$\\
 \end{tabular}
 };

 \path 	(S0) 	edge  node[above]{$*| y_1:=0$} (S1);
	\path 	(Init.190) 	edge  node[right]{$*| x_1:=0$} (S0);
	\path 	(Init.350) 	edge  node[right]{$*| y_1:=0$} (S1);
	
 \path 	(S1)	edge	node[right]{$* | x_2:=0$} (S2);
 \path 	(S2)	edge	node[above]{$* | y_2:=0$} (S3);
 \path  (S3)	edge	node[right=0.2cm,align=left]{$* | x_3:=0$ \\ $\ldots$} (S4);
	\path  (S4)	edge	node[above]{$* | y_m:=0$} (S5);
	\draw [->] (S5) -- ++(0,-0.7cm)  -- node[below]{$* | x_1:=0$} ++(-5.5cm,0cm) |- (S0.west);
\end{tikzpicture}

\caption{The timed automaton for the generator part of $\Diamond_{I} a$ for motion planning.  $2m$ is the number of clocks required to store the states of the timed eventually ($\Diamond_I$) operator. It is computed based on the interval $I$. Detailed computation and derivation can be found in \cite{Maler2006a}. $x_i'$ represents the rate of the clock $x_i$. By setting the rate to be 0, we essentially deactivate the clock. The \lq{$*$}\rq ~ symbol means that there is no value for that particular input, output or guard for that state. There are no B\"{u}chi states since the time is bounded}%
\label{fig:EI_GEN}%
\end{figure}
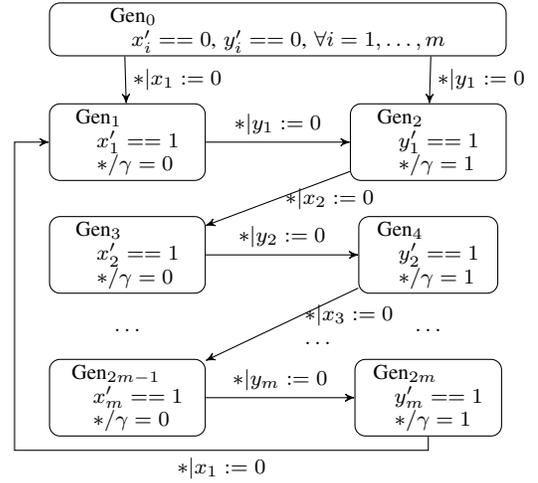

\begin{figure}%
\centering
\begin{tikzpicture}[->,>=stealth']

 \node[normal,
	text width=1.5cm] (S0) 
 {\begin{tabular}{l}
  $\text{Chk}_1$\\
	$y_1\leq b$\\
	$\bar{p}/*$\\
 \end{tabular}};

 \node[normal,
	text width=1.5cm,
	yshift=1.5cm] (Init0) 
 {\begin{tabular}{l}
  $\text{Chk}_{00}$\\
	$x_1 \leq a$\\
 \end{tabular}};

 \node[normal,
	text width=1.5cm,
	right of=Init0,
	node distance=4.5cm] (Init1) 
 {\begin{tabular}{l}
  $\text{Chk}_{01}$\\
	$y_1 \leq a$\\
 \end{tabular}};
  
 \node[normal,    	
  right of=S0, 	
  node distance=3cm, 	
  anchor=center,
	text width=1.5cm] (S1) 	
 {%
 \begin{tabular}{l} 	
  $\text{Chk}_2$\\
	$x_2\leq a$\\
	$p/*$\\
 \end{tabular}
 };

 \node[normal,    	
  right of=S1, 	
  node distance=3cm, 	
  anchor=center,
	text width=1.75cm] (S6) 	
 {%
 \begin{tabular}{l} 	
  $\text{Chk}_3$\\
	$z<b-a$\\
	\,\& $x_2\leq a$\\
	$\bar{p}/*$\\
 \end{tabular}
 };
 
 \node[normal,
  below of=S0,
  yshift=-0.7cm,
  anchor=center,
  text width=1.5cm] (S2) 
 {%
 \begin{tabular}{l}
  $\text{Chk}_4$\\
	$y_2\leq b$\\
	$\bar{p}/*$\\
 \end{tabular}
 };

 \node[normal,
  right of=S2,
  node distance=3cm,
  anchor=center,
	text width=1.5cm] (S3) 
 {%
 \begin{tabular}{l}
  $\text{Chk}_5$\\
	$x_3\leq a$\\
	$p/*$\\
 \end{tabular}
 };
 \node[normal,
  right of=S3,
  node distance=3cm,
  anchor=center,
	text width=1.75cm] (S7) 
 {%
 \begin{tabular}{l}
  $\text{Chk}_6$\\
	$z<b-a$\\
	\,\& $x_3\leq a$\\
	$\bar{p}/*$\\
 \end{tabular}
 };

\node[state] (Sdots) [below of=S2,draw=none,node distance=0.9cm] {$\ldots$};
\node[state] (Sdots2) [below of=S7,draw=none,node distance=0.9cm] {$\ldots$};

 \node[normal,
  below of=Sdots,
  yshift=-0cm,
  anchor=center,
  text width=1.5cm] (S4) 
 {%
 \begin{tabular}{l}
  $\text{Chk}_{3m-2}$\\
	$y_m\leq b$\\
	$\bar{p}/*$\\
 \end{tabular}
 };

 \node[normal,
  right of=S4,
  node distance=3cm,
  anchor=center,
	text width=1.6cm] (S5) 
 {%
 \begin{tabular}{l}
  $\text{Chk}_{3m-1}$\\
	$x_1\leq a$\\
	$p/*$\\
 \end{tabular}
 };
 \node[normal,
  right of=S5,
  node distance=3cm,
  anchor=center,
	text width=1.75cm] (S8) 
 {%
 \begin{tabular}{l}
  $\text{Chk}_{3m}$\\
	$z<b-a$\\
	\,\& $x_1\leq a$\\
	$\bar{p}/*$\\
 \end{tabular}
 };

	\path 	(S0) 	edge  node[above]{$ y_1\geq b|*$} (S1);
	\path 	(S1.5) 	edge  node[above]{$ *|z:=0$} (S6.175);
	\path 	(S6.184) 	edge  node[below]{} (S1.355);
	\path 	(Init0) 	edge  node[left]{$x_1\geq a | *$} (S0);

	\draw [->] 	(Init1.west) 	-|  node[left,yshift=-0.5cm]{$ y_1\geq a | *$} (S1);
	\draw [->] 	(Init1.east) 	-|  node[right,yshift=-0.5cm,align=center]{$y_1\geq a$\\$z:=0$}  (S6);
	
 \path 	(S1)	edge	node[right]{$x_2 \geq a |*$}node[left]{ch!} (S2);
 \path 	(S2)	edge	node[above]{$ y_2\geq b|*$} (S3);
	\path 	(S3.5) 	edge  node[above]{$ *|z:=0$} (S7.175);
	\path 	(S7.184) 	edge  node[below]{} (S3.355);

 \path  (S3)	edge	node[right=0.2cm,align=left]{$x_3\geq a | *$ \\ $\ldots$} (S4);
	\path  (S4)	edge	node[above]{$y_m\geq b | *$} (S5);
	
	\path 	(S5.5) 	edge  node[above]{$ *|z:=0$} (S8.175);
	\path 	(S8.184) 	edge  node[below]{} (S5.355);
	
	\draw [->] (S5) -- ++(0,-0.7cm)  -- node[below]{$ x_1 \geq a | *$} ++(-4.1cm,0cm) |- (S0.west);
\end{tikzpicture}

\caption{The timed automaton for the checker part of $\Diamond_{I} a$ for motion planning. $2m$ is the number of clocks required for the timed eventually ($\Diamond_I$) operator. There are no B\"{u}chi states since the time is bounded}%
\label{fig:EI_CHK}%
\end{figure}
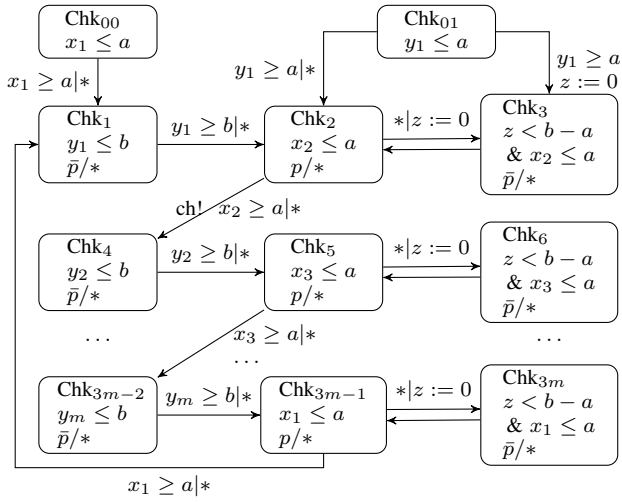

The IOTA for timed eventually ($\Diamond_{I} a$) is decomposed into two automata, the generator generates predictions of the future outputs of the system, while the checker verifies that the generated outputs actually fit the inputs. Detailed derivations and verifications of the models can be found in \cite{Maler2006a}. The composition between them is achieved through the shared clock variables. Additional synchronization (`ch!') is added in our case to determine the final satisfaction condition for the control synthesis. A finite time trajectory satisfies the MITL, when the output signal of the generator automaton (Fig. \ref{fig:EI_GEN}) includes a pair of raising edge and falling edge verified by the checker automaton. The transition from $\text{Chk}_2$ to $\text{Chk}_4$ (Fig. \ref{fig:EI_CHK}) marks the exact time when such falling edge is verified. This guarantees that the time trajectory before the synchronization is a finite time trajectory that satisfies the MITL. 

The composition of IOTA based on logic trees such as that of Fig. \ref{fig:tree} is defined similar to \cite{Maler2006a} with some modifications to handle cases when logic nodes have two children, for example the until and conjunction operators.

\begin{df}[I/O Composition] \hfill \\
\textit{
Let $\mathcal{A}^1_1 =(\Sigma^1_1, Q^1_1, \Gamma^1_1, \mathcal{C}^1_1, \lambda^1_1, \gamma^1_1, I^1_1, \Delta^1_1, {q^1_1}_0, F^1_1)$, $\mathcal{A}^1_2 =(\Sigma^1_2, Q^1_2, \Gamma^1_2, \mathcal{C}^1_2, \lambda^1_2, \gamma^1_2, I^1_2, \Delta^1_2, {q^1_2}_0, F^1_2)$ be the input sides of the automaton. If there is only one, then $\mathcal{A}^1_1$ is used. Let $\mathcal{A}^2 =(\Sigma^2, Q^2, \Gamma^2, \mathcal{C}^2, \lambda^2, \gamma^2, I^2, \Delta^2, q_0^2, F^2)$ be the output side of the automaton. Because of the input output relationship between them, they should satisfies the condition that $[\Gamma^1_1, \Gamma^1_2] = \Sigma^2$. The composition is an new IOTA such that,
\[
\mathcal{A} =(\mathcal{A}^1_1, \mathcal{A}^1_2) \otimes (\mathcal{A}^2) = ([\Sigma_1^1,\Sigma_2^1], Q, \Gamma^2, \mathcal{C}, \lambda, \gamma, I, \Delta, q_0, F)
\]
where
\begin{align*}
Q=&\{(q^1_1,q^1_2,q^2) \in Q^1_1 \times Q^1_2 \times Q^2,  \\
& s.t. (\gamma_1^1(q_1^1),\gamma_2^1(q_2^1)) = \lambda^2(q^2)\}
\end{align*}
$\mathcal{C} = (\mathcal{C}^1_1 \cup \mathcal{C}^1_2 \cup \mathcal{C}^2)$, $\lambda (q^1_1,q^1_2,q^2) = [\lambda^1_1(q^1_1),\lambda^1_2(q^1_2)]$, $I_{(q^1_1,q^1_2,q^2)} = I^1_{(q^1_1,q^1_2)} \cap I^2_{q^2}$, $q_0 = ({q^1_1}_0, {q^1_2}_0, q_0^2)$ and $F = F^1_1 \cap F^1_2 \cup F^2$.
}
\end{df}

\begin{figure*}%
\centering
\includegraphics[width=6.6in]{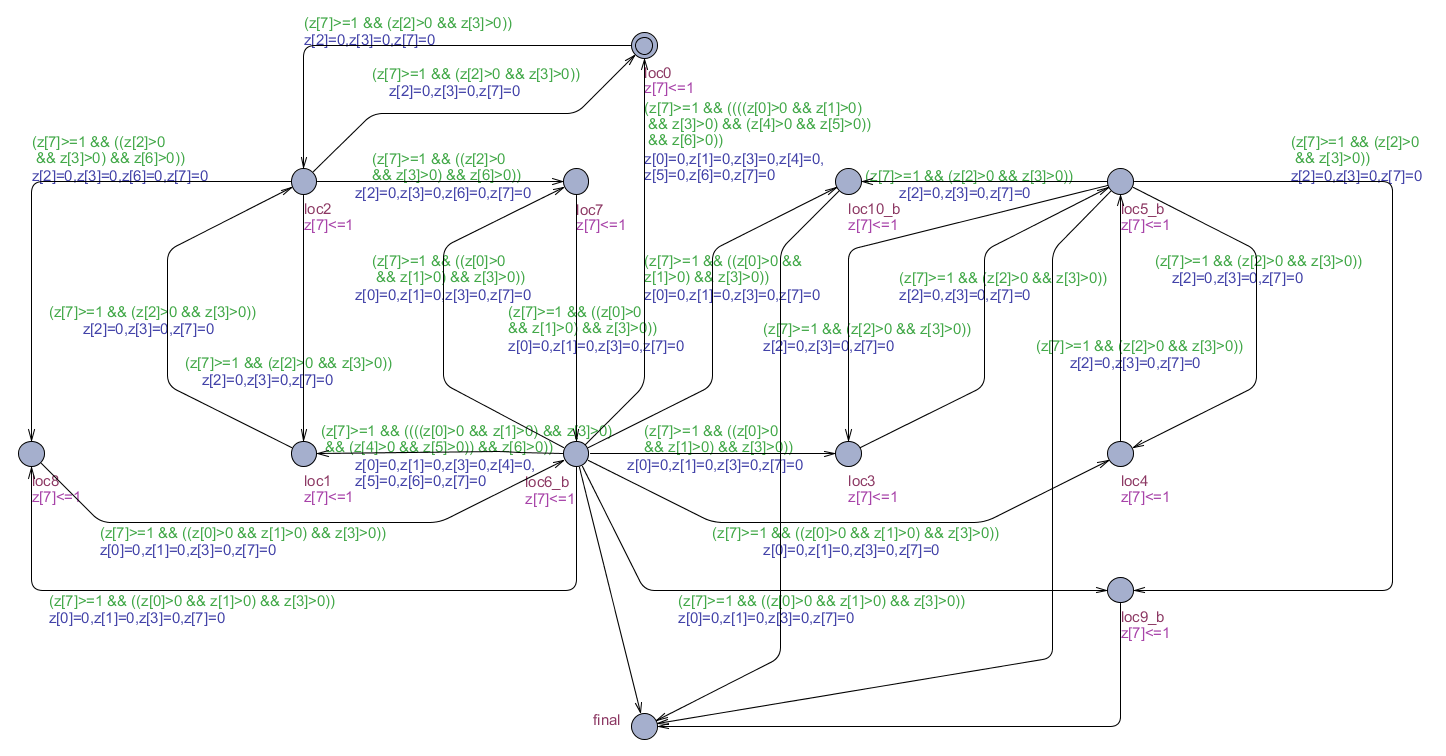}%
\caption{The Resulting timed automaton in UPPAAL of $\phi_1$. The purple colored texts under the state names represent $I$. The green colored texts along the edges represent guard conditions, while the blue ones represent clock resets. The B\"{u}chi accepting states are represented by a subscript b in state names.  }%
\label{fig:until}%
\end{figure*}

\subsection{Path Synthesis using UPPAAL}
The overall path synthesis framework is summarized as following,
\begin{itemize}
	\item First, the robot and the environments are abstracted to a timed automaton (TA) $\mathcal{T}_\text{map}$ using cell decomposition, and the time to navigate from one cell to another is estimated based on the robot's dynamics. For example Fig. \ref{fig:map}.
	\item Second, MITL formula is translated to IOTA $\mathcal{A}$ using method described in previous section.
	\item IOTA  $\mathcal{A}$ is then taken product with the TA $\mathcal{T}_\text{map}$ using the location label. For instance $pos1:B$ in Fig. \ref{fig:map} will be taken product with all states in IOTA that do not satisfy the predicate $p(a)$ but satisfies $p(b)$.
	\item The resulting timed automata are then automatically transformed to an UPPAAL \cite{uppaal} model with additional satisfaction condition verifier. An initial state is chosen so that the output at that state is 1. Any finite trajectory which initiated from that state and satisfying the following conditions will satisfy the MITL specification. Firstly, it has to visit at least one of the B\"{u}chi accepting states, and secondly, it has to meet the acceptance condition for the timed eventually operator. To perform such a search in UPPAAL, a final state is added to allow transitions from any B\"{u}chi accepting state to itself. A verification automaton is created to check the finite acceptance conditions for every timed eventually operator.
	\item An optimal timed path is then synthesized using the UPPAAL verification tool.
\end{itemize}
The implementation of the first and the second step is based on parsing and simplification functions of ltl2ba tool \cite{gastin_fast_2001} with additional capabilities to generate IOTA. We then use the generated IOTA to autogenerate a python script which constructs the UPPAAL model automatically through PyUPPAAL, a python interface to create and layout models for UPPAAL. The complete set of tools\footnote{The tool is available on \url{https://github.com/yzh89/MITL2Timed}} is implemented in C to optimize speed.

\section{Case Study and Discussion} \label{sec:case}
We demonstrate our framework for a simple environment and for some typical temporal logic formulas. Although our tool is not limited by the complexity of the environment, we use a simple environment to make the resulting timed automaton easy to visualize. Let us consider the timed automaton from the abstraction in Fig. \ref{fig:map} and the LTL formula is given as the following,
\[\phi_1 = (\neg A \mathbf{U} B) \wedge (\Diamond A).
\]
This specification requires the robot to visit the area $B$ first and eventually visit $A$ also. The resulting automaton based on the methods in the previous section is as shown in Fig. \ref{fig:until}. Each state corresponds to a product state between a state in $\mathcal{T}_\text{map}$ and a state in IOTA $\mathcal{A}$. The B\"{u}chi accepting states are indicated by an additional \textit{b} in their state names. We obtained the optimal path by first adding a final state and linking every accepting states to it, and then using UPPAAL to find one of the shortest path that satisfies condition ``$E<> final$''. UPPAAL will then compute one fastest path in the timed automaton that goes to final state, if one such exists. If such exists, this feasible path is a finite trajectory that satisfies the specification. 
In this paper, we are more interested in planning a path that satisfies MITL, so finite time trajectory is a valid solution.
The initial states of the automaton is loc0 which is the only state at pos0 that outputs 1. The optimal trajectory is $loc0 \rightarrow loc2 \rightarrow loc7 \rightarrow loc6_b$, in the product automaton. This trajectory means that the optimal way for a robot to satisfy the LTL is to traverse the map in the following order, $pos0 \rightarrow pos1:B \rightarrow pos0 \rightarrow pos3:A$.

\begin{figure*}%
\includegraphics[width=7in]{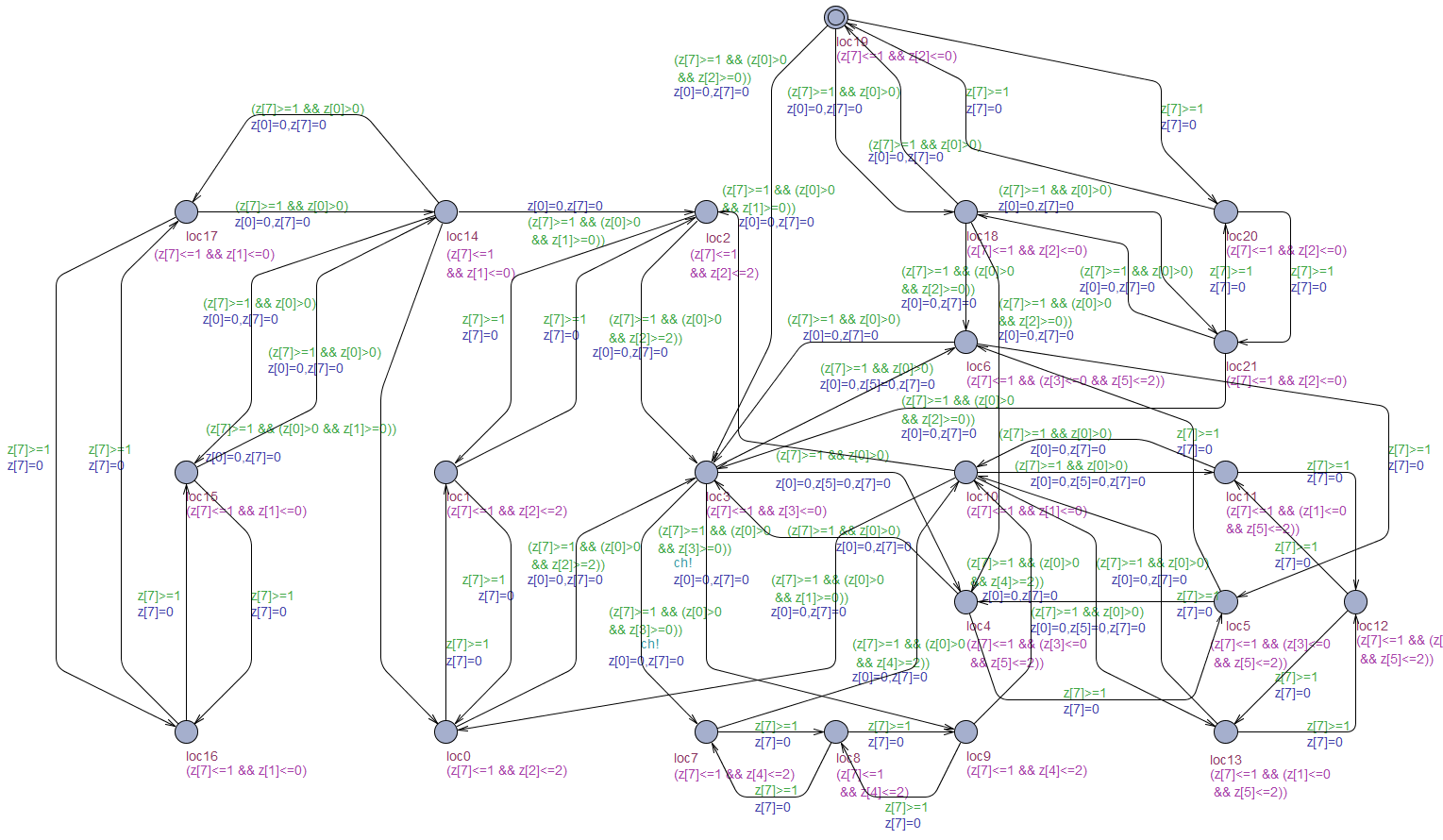}%
\caption{This shows one of the resulting timed automata in UPPAAL of $\phi_2$ corresponding to the checker of timed eventually operator and untimed always. Some of the edges are further annotated by synchronization signal (ch!).}%
\label{fig:always_event_I}%
\end{figure*}

\begin{figure*}[!tb]
    \centering{\subfloat[]{\includegraphics[width=3.5in]{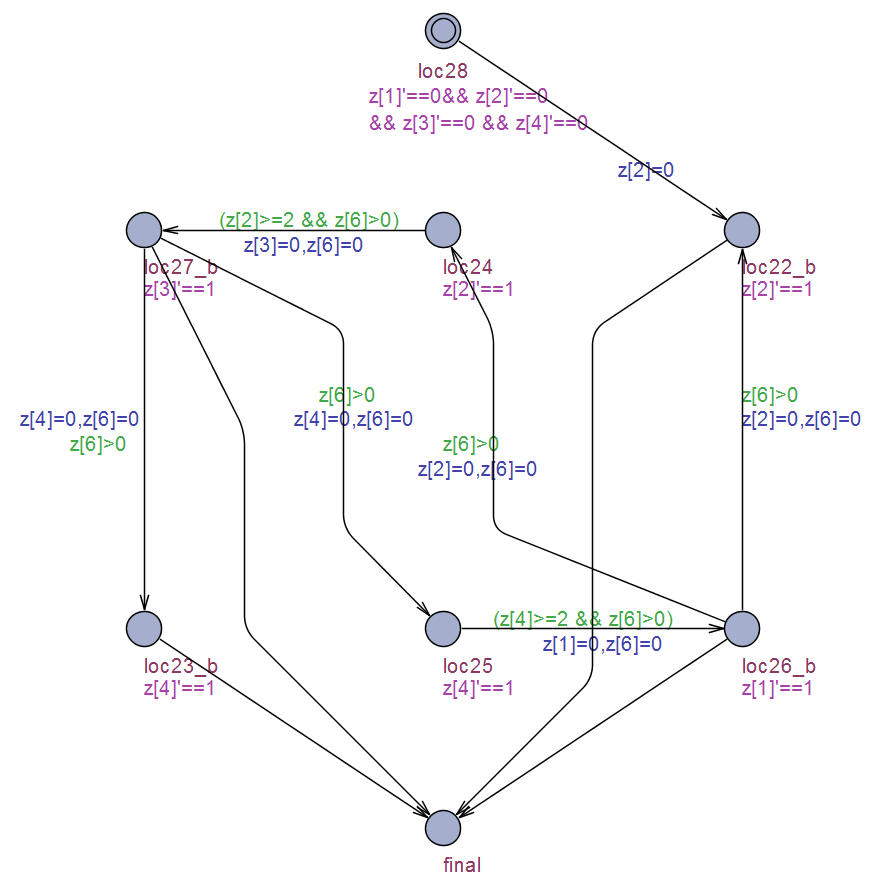}
		\label{fig:always_event_I1}}
		\hfil
		\subfloat[]{\includegraphics[width=1.5in]{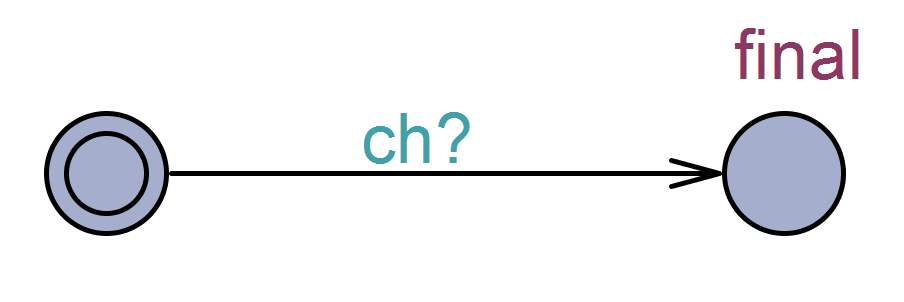}
		\label{fig:always_event_I2}}}
		\caption{Fig. (a) shows the other timed automata of $\phi_2$ corresponding to the generator of timed eventually. Fig. (b) shows the verification timed automaton, that checks if the falling edge of generator is ever detected, i.e. if a synchronization signal (ch!) has happened. This signal marks the end of a full eventually cycle. Similar to the LTL case, we ask UPPAAL to check for us the following property, if there is a trajectory that leads to the final states in (a) and (b). The optimal path in this case is $(loc19,loc28)\rightarrow (loc3,loc28) \rightarrow (loc3,loc22_b)$. The states are products of states of Fig. \ref{fig:always_event_I} and Fig. (a). This path corresponds to $(pos0,t\in[0,1])\rightarrow (pos3:A,t\in[1,2]) \rightarrow (pos0,t\in[2,3])$ in physical space. Repeating this path will satisfy $\phi_2$. }
\label{fig:always_event_I_2}
\end{figure*}

In the second test case, the environment stays the same and the requirement is captured in a MITL formula $\phi_2$
\[
\phi_2 = \Box \Diamond_{[0,2]} A
\]
This requires the robot to perform periodic survey of area \textit{A} every 2s. The resulting timed automata are shown in Fig. \ref{fig:always_event_I} and Fig. \ref{fig:always_event_I_2}. As we discussed earlier, if a synchronization signal (ch!) is sent, the falling edge for output of generator automaton is detected and verified. This marks the end of a finite trajectory that satisfies the MITL constraints. We used the automaton in Fig. \ref{fig:always_event_I_2} (b) to receive such signal. Similar to the LTL case, we ask UPPAAL to find a fastest path that leads to the final states in Fig. \ref{fig:always_event_I_2}(a) and \ref{fig:always_event_I_2}(b) if such exists.

The optimal trajectory in this case is $(loc19,loc28)\rightarrow (loc3,loc28) \rightarrow (loc3,loc22_b)$, which corresponds to $(pos0,t\in[0,1])\rightarrow (pos3:A,t\in[1,2]) \rightarrow (pos0,t\in[2,3])$. Then this trajectory repeats itself.

All the computations are done on a computer with 3.4GHz processor and 8GB memory. Both of the previous examples require very small amount of time $(<0.03s)$. We also tested our implementation against various other complex environments and MITL formulas. The Table \ref{MITLTime} summarizes our results for complex systems and formulas. The map we demonstrated earlier is a 2x2 map (Fig. \ref{fig:map}), we also examine the cases for 4x4 and 8x8 grid maps. The used temporal logic formulas are listed below. The time intervals in the formula is scaled accordingly to the map size.
\[
\phi_3 = \Diamond_{[0,4]} A \wedge \Diamond_{[0,4]} B
\]
\[\phi_4 = \Diamond_{[2,4]} A \wedge \Diamond_{[0,2]} B\]

\begin{table}
\begin{center}
\caption{Computation Time for typical MITL formula}
\label{MITLTime}
\begin{tabular}{|c|c|c|c|c|}
\hline
MITL& Map  & Transformation & Num of Timed & Synthesis\\
Formula& Grid &Time & Automata Transitions & Time \\\hline
$\phi_1$ &2x2 & $<0.001s$ & 22 & 0.016s\\ \hline
$\phi_2$ &2x2 & $0.004s$ & 69 & 0.018s\\ \hline
$\phi_3$ &2x2 & $0.40s$ & 532 & 0.10s\\ \hline
$\phi_4$ &2x2 & $0.46s$ & 681 & 0.12s\\ \hline
$\phi_1$ &4x4 & $0.004s$ & 181 & 0.062s\\ \hline
$\phi_1$ &8x8 & $0.015s$ & 886 & 0.21s\\ \hline
$\phi_2$ &8x8 & $0.015s$ & 1795 & 0.32s\\ \hline
\end{tabular}
\end{center}
\end{table}

It can be seen from the Table \ref{MITLTime} that our algorithm works very well with common MITL formulas and scales satisfactorily with the dimensions of the map.

\section{Continuous Trajectory generation}\label{sec:continuous} 

In this section, we briefly talk about generating a continuous trajectory from the discrete motion plan obtained from the timed automaton. Let us consider the nonholonomic dynamics of a unicycle car as given in (\ref{eq:dynamics}).

\begin{equation} \label{eq:dynamics}
 \dot{\begin{bmatrix}
x\\y\\ \theta
\end{bmatrix}} =u\begin{bmatrix}
\cos\theta \\ \sin\theta \\0
\end{bmatrix}+ \omega \begin{bmatrix}
0\\0\\1
\end{bmatrix}
\end{equation}
where $\omega$ and $u$ are the control inputs. It should be noted that the above nonholonomic dynamics is controllable and we assume no constraints on the control inputs at this point. 
The above sections provide the sequence of cells to be visited in the grid like environment (Fig. \ref{fig:traj}). 
\begin{figure}
\centering
\includegraphics[width=3in]{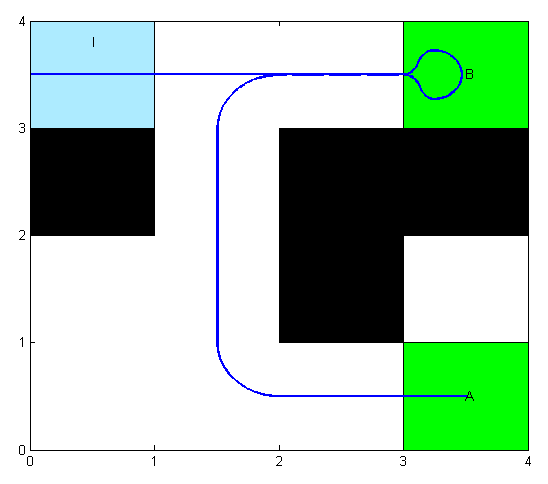} 
\caption{Workspace and the continuous trajectory for the specification $\phi_1$. The initial location is the top-left corner cell (I).}
\label{fig:traj}
\end{figure}

The output of the timed automaton are treated as the time-stamped way points for the robot to move. We have to assure that the robot moves from one way point to the next with the given initial and final time and at the same time, the trajectory should remain within the associated cells. 

Since our environment is decomposed in rectangular cells, the robot will only move forward, turn right, turn left and make a U-turn. We synthesize a controller that can make the robot to perform these elementary motion segments within the given time.

 For moving forward the input $\omega$ is chosen to be $0$ and the velocity $u$ is tuned so that the robot reaches the final position in time. For turning left and turning right $\omega$ is chosen to take positive and negative values respectively so that a circular arc is traversed. Similarly the U-turn is also implemented so that the robot performs the U-turn within a single cell. 
 
 Let us denote the state of the system at time $t$ by the pair $(q,t)$ i.e. $x(t)=q_1,~y(t)=q_2$ and $\theta(t)=q_3$ where $q=[q_1,~q_2,~q_3]$. Then we have the following lemma on the optimality of the 
 control inputs.
 
 \begin{lm} \label{lem:opt}
  \textit{ If $\bar u(t)$ and $\bar \omega (t)$, $t \in [0,1]$ is a pair of control inputs s.t. the dynamics moves from the state $(q_0,0)$ to $(q_1,1)$, then $u(t_0+t)=\frac1\lambda \bar u(
 \frac t\lambda)$ and $\omega(t_0+t)=\frac1\lambda \bar \omega( \frac t\lambda)$ move the system from $(q_0,t_0)$ to $(q_1,t_0+\lambda)$ for any $\lambda >0$. \\
 Moreover, if $\bar u$ and $\bar \omega$ move the system optimally, i.e.
 \begin{equation}
  J(\bar u,\bar \omega)=\min_{u(\cdot),w(\cdot)} \int_{0}^{1} [r_1 u^2(t)+r_2 w^2(t)]dt
 \end{equation}
 then $u$ and $\omega$ given above are also optimal for moving the system from $(q_0,t_0)$ to $(q_1, t_0+\lambda)$, i.e.
 \begin{equation}
 J_1(u,\omega)=\min_{u_1(\cdot),w_1(\cdot)} \int_{t_0}^{t_0+\lambda}[r_1 u^2_1(t)+r_2 w^2_1(t)]dt. 
 \end{equation}
}
 \end{lm}

\begin{proof}
Let us first denote 
\[
G(q)=\begin{bmatrix}
                            \cos(\theta(t)) && 0\\ \sin(\theta(t)) && 0\\0 && 1
                           \end{bmatrix}
													\]
 where $q=[x(t), y(t),\theta(t)]$. Therefore, dynamics (\ref{eq:dynamics}) can be written as $\dot q=G(q) \begin{bmatrix}
                                                                                                       u \\\omega
                                                                                                      \end{bmatrix}.
$
Let us now consider $\bar q(t)=[x(t_0+\lambda t),~y(t_0+\lambda t),~\theta(t_0+\lambda t)]$.   
Therefore, $\dot{\bar q} =\lambda G(\bar q)\begin{bmatrix} u(t_0+\lambda t) \\ \omega (t_0+\lambda t) \end{bmatrix}$. 
Using the definition of $u$ and $\omega$ in the lemma, we get $\dot{\bar q}=G(\bar q)
\begin{bmatrix}
\bar u \\\bar \omega
\end{bmatrix}$
By the hypothesis of the lemma, $\bar{u}$ and $\bar \omega$ move the system from $(q_0,0)$ to $(q_1,1)$ i.e. from $[x(t_0),y(t_0),\theta(t_0)]=q_0$ to $q_1=\bar q(1)=[x(t_0+\lambda),y(t_0+\lambda),\theta(t_0+\lambda)]$.\\

For optimality, let the proposed $u,~\omega$ be not optimal and $u^*$ and $\omega^*$ are optimal ones i.e. 
\begin{equation}
\int_{t_0}^{t_0+\lambda}[r_1{u^*}^2(t)+r_2{\omega^*(t)}^2]dt \le \int_{t_0}^{t_0+\lambda}[r_1u^2(t)+r_2\omega^2(t)]dt
\label{eq:opt1}
\end{equation}
Now let us construct $\bar u^*(t)=\lambda u^*(t_0+\lambda t)$ and $\bar \omega^*(t)=\lambda \omega^*(t_0+\lambda t)$.

Therefore from (\ref{eq:opt1}),
\begin{multline*}
\int_{0}^{1}[r_1{u^*}^2(t_0+\lambda s)+r_2{\omega^*}^2(t_0 +\lambda s)]ds \\
\leq \int_{0}^{1}[r_1u^2(t_0+\lambda s)+r_2\omega^2(t_0+\lambda s)]ds
\end{multline*}

\begin{equation} \int_{0}^{1}[r_1 {\bar {u^*}}^2(s)+r_2{\bar {\omega^*}}^2(s)]ds \le \int_{0}^{1}[r_1{\bar u}^2(s)+r_2\bar \omega^2(s)]ds \label{Eqn:ineq1} \end{equation}

But, by the hypothesis, $\bar u$ and $\bar \omega$ are optimal and hence 
\begin{equation} \label{Eqn:ineq2}
 \int_{0}^{1}[r_1 {\bar {u^*}}^2(s)+r_2{\bar {\omega^*}}^2(s)]ds \ge \int_{0}^{1}[r_1{\bar u}^2(s)+r_2\bar \omega^2(s)]ds
\end{equation}
Combining (\ref{Eqn:ineq1}) and (\ref{Eqn:ineq2}) we get,
\begin{equation}
 \int_{0}^{1}[r_1 {\bar {u^*}}^2(s)+r_2{\bar {\omega^*}}^2(s)]ds = \int_{0}^{1}[r_1{\bar u}^2(s)+r_2\bar \omega^2(s)]ds
\end{equation}

After changing the dummy variables inside integration again, one can obtain 

\begin{equation}
 \int_{t_0}^{t_0+\lambda}[r_1 { {u^*}}^2(s)+r_2{{\omega^*}}^2(s)]ds = \int_{t_0}^{t_0+\lambda}[r_1{u}^2(s)+r_2\omega^2(s)]ds
\end{equation}

Hence the proposed $u$ and $\omega$ are optimal whenever $\bar u$ and $\bar \omega$ are optimal.
  \end{proof}
 \begin{rem}
  Lemma \ref{lem:opt} states that if the controls for elementary motions from initial time $0$ to final time $1$ are synthesized, then by properly scaling and shifting in the time and scaling the magnitude, controls for any movement from any initial time to any final time can be synthesized without further solving any optimization problem.
 \end{rem}

\section{Conclusion} \label{sec:conclusion}
In this paper, we have presented a timed automaton based approach to generate a discrete plan for the robot to perform temporal tasks with finite time constraints. We implemented the algorithm in an efficient and generic way so that it can translate the time constraints and temporal specifications to timed automaton models in UPPAAL and synthesize the path accordingly. We then demonstrated our algorithm in grid type environments with different MITL formulas. We have considered grid type of environment for our case studies, but it can be generalized to most of the motion planning problems when the environment can be decomposed into cells. We also provide a brief overview of how an optimal continuous trajectory can be generated from the discrete plan. For future works, we are considering to extend the work to include dynamic obstacles as well as for multiagent system. 



\bibliographystyle{IEEEtran}
\bibliography{IEEEabrv,bib}

\end{document}